\documentclass[11pt]{amsart}
\usepackage[margin=1.2in]{geometry}
\usepackage{amsthm,amsmath,amsfonts,amssymb,euscript,hyperref,graphics,color,slashed,mathrsfs,graphicx,mathtools,enumerate,cite}
\hypersetup{colorlinks, citecolor=blue, urlcolor=blue, linkcolor=blue}
\usepackage[english]{babel}
\usepackage{bm}
\usepackage{slashed}

\allowdisplaybreaks

\numberwithin{equation}{section}

\newtheorem{Theorem}{Theorem}
\newtheorem{Lemma}[Theorem]{Lemma}

\newtheorem{Corollary}[Theorem]{Corollary}

\newtheorem{Remark}{Remark}

\title[Rigidity for Minkowski spacetime]{Remark on the nonlinear stability of Minkowski spacetime: a rigidity theorem
	}

\author[Jin Jia]{Jin Jia}
\address{Department of Mathematics and Yau Mathematical Sciences Center, Tsinghua University\\ Beijing, China}
\email{jiaj18@mails.tsinghua.edu.cn}

\author[Pin Yu]{Pin Yu}
\address{Department of Mathematics and Yau Mathematical Sciences Center, Tsinghua University\\ Beijing, China}
\email{yupin@mail.tsinghua.edu.cn}

\begin{document}
\begin{abstract}
In the framework of the nonlinear stability of Minkowski spacetime, we show that if the radiation field of the curvature tensor vanishes, the spacetime must be flat.
\end{abstract}

\maketitle
\tableofcontents
\section{Introduction}

In general relativity, the objects under consideration are given by $(M,g)$ where $M$ is a 4-dimensional smooth manifold and $g$ is a Lorentzian metric on $M$. Let ${\rm Ric}(g)$ and $R(g)$ be the Ricci and scalar curvatures of the metric $g$ respectively. The Einstein vacuum equations are given by
\begin{equation*}
{\rm Ric}(g) - \frac{1}{2} R(g)\cdot g = 0.
\end{equation*}
Equivalently, it requires that $(M,g)$ is Ricci flat, i.e.,
 \begin{equation*}
{\rm Ric}(g) =0.
\end{equation*}

Let $\Sigma$ be a three dimensional differentiable manifold. A Cauchy initial datum for vacuum Einstein field equations on $\Sigma$ consists of a Riemannian metric $\overline{g}$ and a symmetric two tensor $\overline{k}$ (as the second fundamental form) subject to the following  constraint equations:
\begin{equation}\label{constraint}
\begin{cases}
R(\overline{g}) -|\overline{k}|^2 + \left({\rm tr}_{\overline{g}} \overline{k}\right)^2 &= 0,\\
{\rm div}_{\overline{g}} \overline{k} - \overline{\nabla}  \left({\rm tr}_{\overline{g}} \overline{k}\right) &=0,
\end{cases}
\end{equation}
where $R(\overline{g})$ is the scalar curvature of the metric $\overline{g}$ and $\overline{\nabla}$ is the Levi-Civita connection associated to $\overline{g}$. The Cauchy problem for vacuum Einstein equations with the given data $(\Sigma, \overline{g},\overline{k})$ is to construct a spacetime $(M, g)$  so that there exists an isometric embedding  $\Sigma \subset M$ in such a way  that $\overline{g}$ coincides with the restriction of $g$ to  $\Sigma$ and $\overline{k}$ coincides with the second fundamental form of $\Sigma$.

The local existence of the Einstein equations in harmonic gaugehas already been established by Choquet-Bruhat in \cite{Choquet-Bruhat} in 1952. Later on, Choquet-Bruhat and Geroch \cite{Choquet-Bruhat-Geroch} showed  the existence of a maximal Cauchy development to the Einstein equations for sufficiently smooth initial data. These results correspond to the local existence theory in the language of PDEs and the global behavior of the Einstein equations is much more challenging. Nevertheless, in 1993 Christodoulou and Klainerman have made a breakthrough in this direction: they proved the nonlinear stability of Minkowski spacetime \cite{Christodoulou-Klainerman}. Roughly speaking, they have  constructed a global solution to the Einstein vacuum equations from the data sufficiently close to a spacelike hyperplane in the Minkowski and they also showed that the solution asymptotically approach to the Minkowski space for large time. The proof of Christodoulou and Klainerman is a landmark in the field of mathematical general relativity. Afterwards, the result was extended in many other situations, see \cite{Klainerman-Nicolo1}, \cite{Klainerman-Nicolo2}, \cite{Bieri-Zipser}, \cite{Lindblad-Rodnianski1}, \cite{Lindblad-Rodnianski2}, \cite{Lindblad} and \cite{Speck} for Einstein vacuum equations or Einstein-Maxwell equations, see \cite{BFJ},\cite{Lindblad-Taylor} and \cite{XWang} for the Einstein-Vlasov system, see \cite{Wang}, \cite{LeFloch-Ma} and \cite{Ionescu-Pausader} for the Einstein-Klein-Gordon system.

\subsection{A review on Christodoulou-Klainerman's work}

In the rest of the paper, we require that the initial datum triplet $(\Sigma,\overline{g},\overline{k})$ is \emph{strongly asymptotically flat}. This means that in addition to the constraint equations \eqref{constraint}, $(\Sigma,\overline{g},\overline{k})$ statisfies the following properties: 
\begin{itemize}
\item The three dimensional manifold $\Sigma$ is diffeomorphic to $\mathbb{R}^3$. We use $(x^1,x^2,x^3)$ as the standard coordinate system on $\Sigma$. We also use $r=\sqrt{\left(x^1\right)^2+\left(x^2\right)^2+\left(x^3\right)^2}$ to denote the usual radius function. 
\item For $r\rightarrow \infty$, in the coordinate frame, the induced metric $\overline{g}$ and the second fundamental form $\overline{k}$ satisfy:
\begin{itemize}
\item For all $i,j \in \{1,2,3\}$, we have 
\[\overline{g}_{ij}=\left(1+\frac{2M}{r}\right)\delta_{ij}+O_4(r^{-\frac{3}{2}}).\]
\item For all $i,j \in \{1,2,3\}$, we have 
\[\overline{k}_{ij}=O_3(r^{-\frac{5}{2}}).\]
\end{itemize}
In the above expressions, we denote by $O_k(r^a)$ any smooth function $f$ defined on  $\Sigma$ which satisfies 
 \[|\partial^i f|=O(r^{a-i})\]  for any $0\leq i\leq k$ with $\displaystyle |\partial^i
f|=\sum_{i_1+i_2+i_3=i}
|\partial_1^{i_1}\partial_2^{i_2}\partial_3^{i_3}
f|$.
\end{itemize}

We fixed a marked point $x_{(0)} \in \Sigma$. Let $d\left(x_{(0)} , x\right)$ be the geodesic distance on $\Sigma$ between $x_{(0)}$ and $x\in \Sigma$ with respect to the metric $\overline{g}$. This defines a function on $\Sigma$:
\[d_0: \Sigma \rightarrow \mathbb{R}_{\geq 0}, \ \ x\mapsto d(x)=d\left(x_{(0)} , x\right).\]
We also introduce the following symmetric and traceless (with respect to $\overline{g}$) 2-tensor:
\[\overline{B}_{ij}=\overline\epsilon_j{}^{ab}\overline{\nabla}_a\left(\overline{R}_{ib}-\frac{1}{4}\overline{g}_{ib}\overline{R}\right),\]
where $\overline{\epsilon}_{ijk}$ is the volume form of $\overline{g}$, $\overline{R}_{ij}$ and $\overline{R}$ are the Ricci and scalar curvature of the metric $\overline{g}$. We also introduce
\begin{align*}Q\left(x_{(0)}\right)&=\int_{\Sigma} \sum_{l=0}^3 \left(d_0(x)^2+1\right)^{l+1}|\overline{\nabla}^l \overline{k}|^2 dx+ \int_{\Sigma}\sum_{l=0}^1 \left(d_0(x)^2+1\right)^{l+3}|\overline{\nabla}^l \overline{B}|^2 dx\\
&\ \ \ +\sup_{x\in \Sigma}\left\{\left(d_0(x)^2+1\right)^3|{\rm Ric}(\overline{g})(x)|^2\right\}.
\end{align*}
The first version of the theorem of Christodoulou and Klainerman can be stated as follows:
\medskip

\emph{
Given a strongly asymptotically flat maximal \footnote{A maximal initial datum $(\Sigma, \overline{g}, \overline{k})$ satisfies ${\rm tr}_{\overline{g}} \overline{k}=0$.} initial datum $(\Sigma, \overline{g}, \overline{k})$ to the vacuum Einstein equatiions, there exists a positive constant $\varepsilon_0$, so that if $(\Sigma, \overline{g}, \overline{k})$ satisfies the following global smallness assumption:
\[\int_{x_{(0)}\in \Sigma}Q(x_{(0)})<\varepsilon_0,\]
it leads to a unique, globally hyperbolic, smooth and geodesically complete solution of the Einstein vacuum equations foliated by a normal maximal time function $t$, defined for all $t\geq -1$. We denote this spacetime by $(M,g)$ and $t=0$ corresponds to $\Sigma$. Moreover,this development is globally asymptotically flat.
}
\subsection{The geometry of the spacetime in null frames}
Let $(M,g)$ be the space time constructed in \cite{Christodoulou-Klainerman}. For the region where $t\geq 0$, $M$ admits global foliations $\{\Sigma_{t}\}_{t\geq 0}$ given by level sets of $t$ and an exterior foliation $\{C_{u}\}_{u\in \mathbb{R}}$ given by level sets of an optical function $u$. For a fixed time $t$,  when $u$ changes, the following spheres
\[S_{t,u}=\Sigma_{t}\cap C_{u}
\]
gives an exterior foliation of $\Sigma_{t}$. We use $\theta$ to denote the second fundamental form of the embedding $S_{t,u} \hookrightarrow \Sigma_{t}$.

Let $T$ be the future directed unit normal of $\Sigma_t$ and $N$ be the outward unit normal of $S_{t,u}$ in $\Sigma_{t}$. We define the standard null pair $(e_{3},e_{4})$ given by the following formulas:
\[e_{3}=T-N,\ \ e_{4}=T+N.\]
In particular, we have $g(e_{3},e_{4})=-2$. The projection operator $\Pi$, which projects  tensors on ${M}$ to the tensor tangential to $S_{t,u}$, can be computed in terms of the standard null pair:
\[\Pi^{\mu v}=g^{\mu\nu}+\frac{1}{2}\left(e_{3}^{\mu}e_{4}^{\nu}+e_{4}^{\mu}e_{3}^{\nu}\right).
\]
We use  $(e_1,e_2,e_3,e_4)$ to denote a standard null pair on $M$, where $(e_3,e_4)$ is supplemented by $(e_A)_{A = 1, 2}$, a local orthonormal frame field for $S_{t,u}$. In the rest of the paper, we will use capital Roman letters $A,B,\cdots$ to denote indices $\{1,2\}$ and use Greek letters $\mu,\nu,\cdots$ to denote indices $\{1,2,3,4\}$. Repeating indices will be understood by the Einstein's summation convention.

We also introduce the null pair $(L,\underline{L})$ defined by the gradient of the optical function $u$. Therefore, we have
\[L^{\mu}=-g^{\mu\nu}\frac{\partial u}{\partial x^{\nu}}=a^{-1}(T+N),\ \ \underline{L}=a(T-N),\]
where $a$ is the lapse function of $u$ on $\Sigma_t$, i.e., $a=|{\nabla} u|^{-1}$. As a convention, we use $D$ to denote the Levi-Civita connection defined by $g$ and we also use $\nabla$ to denote the induced connection on $\Sigma_t$. 

The corresponding Ricci coefficients for the connection $D$ are given by the following set of equations:
\begin{align*}
 \chi_{AB}&=g(D_{A}e_{3},e_{B}), \  2\xi_{A}=g(D_{4}e_{4},e_{A})=0, \  \eta_{A}= \frac{1}{2}g(D_{3}e_{4},e_{A}), \  \omega=\frac{1}{4}g(D_{4}e_{4},e_{3}),\\
	\underline{\chi}_{AB}&=g(D_{A}e_{4},e_{B}), \  2\underline{\xi}_{A}=g(D_{3}e_{3},e_{A}), \  \underline{\eta}_{A}= \frac{1}{2}g(D_{4}e_{3},e_{A}), \ \omega=\frac{1}{4}g(D_{3}e_{3},e_{4}).
\end{align*}
We also introduce the torsion $\zeta_A =\frac{1}{2}g(D_{A}e_{4},e_{3})$. We can write the connection $D$ in terms of the null frame $(e_\mu)$ and the Ricci coefficients:

\begin{align*}
D_{4}e_{4}&=-2\omega e_{4}, \ D_{3}e_{4}=2\eta_{A}e_{A}+2\underline{\omega}e_{4},\ D_{A}e_{4}=\chi_{AB}e_{B}-\zeta_{A}e_{4}, \\
D_{3}e_{3}&=2\underline{\xi}_{A}e_{A}-2\underline{\omega}e_{3},\	D_{4}e_{3}=2\underline{\eta}_{A}e_{A}+2\omega e_{3},\ D_{A}e_{3}=\underline{\chi}_{AB}e_{B}+\zeta_{A}e_{3},\\
D_{4}e_{A}&=\nabla_{4}e_{A}+\underline{\eta}_{A}e_{4}	+\xi_{A}e_{4}, \ D_{3}e_{A}=\nabla_{3}e_{A}+\eta_{A}e_{3}+\underline{\xi}_{A}e_{4},\\
D_{B}e_{A}&=\slashed{\nabla}_{B}e_{A}+\frac{1}{2}\chi_{AB}e_{3}+\frac{1}{2}\underline{\chi}_{AB}e_{4}.
\end{align*}

In the above formulas, the notations $\nabla_{4}e_{A}$ and $\nabla_{3}e_{A}$ are the projection of $D_4 e_{A}$ and $D_3 e_{A}$ to $S_{t,u}$, e.g., $\Pi(D_4 e_{A})$. Similarly, $\slashed{\nabla}_{B}e_{A}=\pi(D_{B}e_{A})$.
\begin{Remark}
The notations in the current work are different from that in \cite{Christodoulou-Klainerman}. The notations $H$, $\underline{H}$, $Y$, $\underline{Y}$, $Z$, $\underline{Z}$, $\Omega$, $\underline{\Omega}$ and $V$ are replaced by $\chi$, $\underline{\chi}$, $\xi$, $\underline{\xi}$, $\eta$, $\underline{\eta}$, $\omega$, $\underline{\omega}$ and $\zeta$ respectively.
\end{Remark}
We can also write the induced connection $\nabla$ on $\Sigma_{t}$ in terms of the frame $(N,e_1,e_2)$:
\begin{align*}
\nabla_{N}e_{A}&=\slashed{\nabla}e_{A}+a^{-1}(\slashed{\nabla}_{A}a)N, \  \nabla_{A}N=\theta_{AB}e_{B},\ \nabla_{N}N=-a^{-1}(\slashed{\nabla}_{A}a)e_{A},\ \nabla_{B}e_{A}=\slashed{\nabla}_{B}e_{A}-\theta_{AB}N.
\end{align*}
We now relate the spacetime Ricci coefficients to those defined on $\Sigma_t$. We first decompose the second fundamental form $k$ into
\[
	\Xi_{AB}=k_{AB},\ \epsilon_{A}=k_{AN},\ \delta=k_{NN}.
\]
Therefore, if we use $\phi$ to denote the lapse function of the time foliation $\Sigma_t$, i.e., $\phi = \left(-g^{\mu\nu}D_\mu t D_\nu t\right)^{-\frac{1}{2}}$ or equivalently $T= -\phi Dt$,  we have 
\begin{align*}
\chi_{AB}&=\theta_{AB}-\Xi_{AB}, \ \eta_{A}=a^{-1}\nabla_{A}a+\epsilon_{A},\ \omega=\frac{1}{2}(-\phi^{-1}\nabla_{N}\phi+\delta),\\
\underline{\chi}_{AB}&=-\theta_{AB}-\Xi_{AB}, \ \underline{\eta}_{A}=\phi^{-1}\nabla_{A}\phi-\epsilon_{A},\ \underline{\omega}=\frac{1}{2}(\phi^{-1}\nabla_{N}\phi+\delta),
\end{align*}
and

\[\underline{\xi}_{A}=\phi^{-1}\nabla_{A}\phi-a^{-1}\nabla_{A}a,\ \zeta_{A}=\epsilon_{A}.\]

We also use $r$ to the radius of the sphere $S_{t,u}$, i.e., $4\pi r^2$ is equal to the area of $S_{t,u}$. Then, we have 
\[\nabla_{N}r=\frac{r}{2a}\overline{a {\rm tr}\theta}, \ D_{T}r=\phi^{-1}\frac{r}{2}\overline{\phi {\rm tr}\chi},\]
where the notation $\overline{f}$ denotes the mean of the function $f$ on $S_{t,u}$, i.e., $\overline{f}:=\frac{1}{4\pi r^{2}}\int_{S_{t,u}}f d\mu$ with $d\mu$ the volume form of $S_{t,u}$.


Finally, we decompose the curvature tensor in terms of the null frame $(e_\mu)$. Indeed, we can decompose any Weyl type tensor $W$ as follows:
\begin{align*}
\alpha(W)_{AB}&= W(e_A,e_4,e_B,e_4),\ \beta(W)_{A}= \frac{1}{2}W(e_A,e_4,e_3,e_4), \ \rho(W)_{A}=  \frac{1}{4}W(e_3,e_4,e_3,e_4),\\
\underline{\alpha}(W)_{AB}&= W(e_A,e_3,e_B,e_3),\ \underline{\beta}(W)_{A}=  \frac{1}{2}W(e_A,e_3,e_3,e_4), \ \rho(W)_{A}= \frac{1}{4} \,^*W(e_3,e_4,e_3,e_4).
\end{align*}
where $^*W$ is the Hodge dual of $W$. In particular, $\alpha$ and $\underline{\alpha}$ are symmetric traceless $2$-tensors on $S_{t,u}$, $\beta$ and $\underline{\beta}$ are $1$-forms on $S_{t,u}$.  The above $6$ tensors completely determine the Weyl type tensor $W$ in the following way:
\begin{align*}
	W_{A3B3}&=\underline{\alpha}_{AB},W_{A4B4}=\alpha_{AB},  W_{A334}=2\underline{\beta}_{A},W_{A434}=\beta_{A}, W_{3434}=\rho,W_{AB34}=2\sigma\epsilon_{AB},\\
	W_{ABC3}&=\epsilon_{AB}\,^*\underline{\beta}_{C},W_{ABC4}=-\epsilon_{AB}\,^*\beta_{C},\ W_{A3B4}=-\rho\delta_{AB}+\sigma\epsilon_{AB},W_{ABCD}=-\epsilon_{AB}\epsilon_{CD}\rho.
\end{align*}
In the above expressions, $\,^*\alpha$, $\,^*\underline{\alpha}$,$\,^*\beta$ and $\,^*\underline{\beta}$ are the Hodge duals of the corresponding tensors relative to the metric on $S_{t,u}$. We use $\epsilon_{AB}$ to denote the volume $2$-form of $S_{t,u}$.

The Ricci coefficients and the curvature components are related by the null structure equations and null Bianchi equations. We refer to Chapter 7 of \cite{Christodoulou-Klainerman} to the whole system of equations. We only list two of them which appear in the rest of the paper: 

\begin{equation}\label{eq: null structure: 4 chibh}
\slashed{\nabla}_{N}\widehat{\Xi}+\frac{1}{2}{\rm tr}\theta\widehat{\Xi}=\frac{1}{4}(-\underline{\alpha}+\alpha)+\frac{1}{2}\slashed{\nabla}\widehat{\otimes}\zeta+\frac{3}{2}\delta\widehat{\theta}+(a^{-1}\nabla a)\widehat{\otimes}\zeta.\end{equation}


\begin{equation}\label{eq: null bianchi: 4 alphab}
 	\slashed{\nabla}_{4}\underline{\alpha}+\frac{1}{2}{\rm tr}\chi\underline{\alpha}=-\slashed{\nabla}\widehat{\otimes}\underline{\beta}+4\omega\underline{\alpha}-3(\underline{\widehat{\chi}}\rho-\,^*\underline{\widehat{\chi}}\sigma)+(\zeta-4\underline{\eta})\widehat{\otimes}\underline{\beta}.
\end{equation}

In the above expressions, ${\rm tr}$ means the trace part of the tensor (with respect to the induced metric on $S_{t,u}$) and ~$\widehat{ }$~ means the traceless part of the tensor.

\subsection{The Bel-Robinson tensors and energy identities}
Given a Weyl field $W$, the associated Bel-Robinson tensor $Q(W)$ is defined as 
\begin{align}
	Q(W)_{\alpha\beta\gamma\delta}=W_{\alpha\rho\gamma\sigma}W_{\beta}{}^{\rho}{}_{\delta}{}^{\sigma}+\,^*W_{\alpha\rho\gamma\delta}\,^*W_{\beta}{}^{\rho}{}_{\delta}{}^{\sigma}.
\end{align}
This is a symmetric and traceless $4$-tensor. It is also \emph{positive} in the following sense: for any causal future directed vector fields $X_1,\cdots, X_4$, we have
\[Q(W)(X_{1},X_{2},X_{3},X_{4})\geq 0.\]

If $W=R$ is the Weyl tensor (for a vacuum spacetime), $Q(W)$ is divergence free, i.e.,
\[D^{\alpha}Q(R)_{\alpha\beta\gamma\delta}=0.\]
We refer to Chapter 7 of \cite{Christodoulou-Klainerman} to the proof of the above properties. In the rest of paper, we will only use $Q(R)$ and we denote it by $Q$.

One of the main applications of the Bel-Robinson tensors is to derive energy estimates for curvature tensors. Given $3$ causal future directed vector fields $X$,$Y$ and $Z$, we consider the following current  $1$-form 
\[P_{\alpha}=Q_{\alpha\beta\gamma\delta}X^{\beta}Y^{\gamma}Z^{\delta}.\]
Therefore,
\[
	D^{\alpha}P_{\alpha}=\frac{1}{2}Q^{\alpha\beta\gamma\delta}({}^{(X)}\pi_{\alpha\beta}Y_{\gamma}Z_{\delta}+{}^{(Y)}\pi_{\alpha\beta}X_{\gamma}Z_{\delta}+{}^{(Z)}\pi_{\alpha\beta}X_{\gamma}Y_{\delta}),
\]
where ${}^{(X)}\pi_{\alpha\beta}=\mathcal{L}_X g$ is the corresponding deformation tensor. We can integrate this identity on the region foliated by $\Sigma_{t}$ where $t\in [t_{1},t_{2}]$. According to the Stokes theorem, we have
 \begin{equation}\label{basic energy estimates}
 \begin{split}
 &\int_{\Sigma_{t_{2}}}Q(W)(X,Y,Z,T)-\int_{\Sigma_{t_{1}}}Q(W)(X,Y,Z,T)\\
 =&
 	\frac{1}{2}\int_{t_{1}}^{t_{2}}\phi \left(\int_{\Sigma_{t}}Q^{\alpha\beta\gamma\delta}\,{}^{(X)}\pi_{\alpha\beta}Y_{\gamma}Z_{\delta}+Q^{\alpha\beta\gamma\delta}\,{}^{(Y)}\pi_{\alpha\beta}X_{\gamma}Z_{\delta}+Q^{\alpha\beta\gamma\delta}\,{}^{(Z)}\pi_{\alpha\beta}X_{\gamma}Y_{\delta}\right)dt.
	\end{split}
\end{equation}
 In the above expressions, $\int_{\Sigma_t}$ means the integration on $\Sigma_t$ with respect to the induced metric.
 
 We will take $X$, $Y$ and $Z$ to be one the following vectors:
 \[T=\frac{1}{2}(e_{3}+e_{4}), \ \ K=\frac{1}{2}\left(u^{2}e_{3}+(2r-u)^{2}e_{4}\right).\]
The null components of $\,^{(T)}\pi$ are listed as follows:
\begin{equation}\label{deformation: T}
\begin{split}
	&{}^{(T)}\pi_{44}=-2\omega,\  \ {}^{(T)}\pi_{34}=2\delta,\ \ {}^{(T)}\pi_{33}=-2\underline{\omega},\ {}^{(T)}\pi_{AB}=-2\Xi_{AB},\\&
	{}^{(T)}\pi_{A3}=2\epsilon_{A}+\phi^{-1}\slashed{\nabla}_{A}\phi,\ \ {}^{(T)}\pi_{A4}=-2\epsilon_{A}+\phi^{-1}\slashed{\nabla}_{A}\phi.
\end{split}
\end{equation}
The null components of $\,^{(K)}\pi$ are listed as follows:
\begin{equation}\label{deformation: K}
\begin{split}	
^{(K)}\widehat{\pi}_{44}&=-4u^{2}\omega,\ ^{(K)}\widehat{\pi}_{4A}=u^{2}(\underline{\eta}-\zeta_{A}), \ ^{(K)}\widehat{\pi}_{AB}=u^{2}\widehat{\underline{\chi}}_{AB}+(2r-u)^{2}\widehat{\chi}_{AB}+\frac{1}{2}^{(K)}\widehat{\pi}_{34}\delta_{AB}\\
		^{(K)}\widehat{\pi}_{33}&=-8(2r-u)(D_{3}r+a^{-1})-4(2r-u)^{2}\underline{\omega},\ ^{(K)}\widehat{\pi}_{3A}=u^{2}\underline{\xi}_{A}+(2r-u)^{2}(\eta_{A}+\zeta_{A}),\\
	^{(K)}\widehat{\pi}_{34}&=2u\left(\frac{1}{a}-1\right)-2(2r-u)(D_{4}r-1)+\frac{1}{2}u^{2}\left({\rm tr}\underline{\chi}+\frac{2}{r}\right)+\frac{1}{2}(2r-u)^{2}\left({\rm tr}\chi-\frac{2}{r}\right)\\
	& \ \ +u^{2}\underline{\omega}+(2r-u)^{2}\omega,
\end{split}
\end{equation}
where
\begin{align*}
	D_{4}r=\phi^{-1}\frac{r}{2}\overline{\phi {\rm tr}\chi},\ \ D_{3}r=a^{-1}\frac{r}{2}\overline{a{\rm tr}\underline{\chi}}+\frac{r}{2}\left(\phi^{-1}\overline{\phi {\rm tr}\chi}-a^{-1}\overline{a{\rm tr}\chi}\right)
\end{align*}
We will take $(X,Y, Z) =(T,T,T)$ or $(T,T,\overline{K})$ where $\overline{K}=K+T$. Thus, we compute
\begin{equation}\label{eq: computation for Q}
\begin{split}
	Q(W)(T,T,T,T)&=|\alpha|^{2}+|\underline{\alpha}|^{2}+|\beta|^{2}+|\underline{\beta}|^{2}+|\rho|^{2}+|\sigma|^{2},\\ 
	Q(W)(\overline{K},T,T,T)&=\tau_{-}^{2}|\underline{\alpha}|^{2}+\tau_{+}^{2}(|\alpha|^{2}+|\beta|^{2}+|\underline{\beta}|^{2}+|\rho|^{2}+|\sigma|^{2})
\end{split}
\end{equation}
where $\tau_{-}^{2}=1+u^{2}$ and $\tau_{+}^{2}=1+(2r-u)^{2}$. In particular, we have the following estimate:
\begin{equation}\label{eq: bound Q by conformal Q}
	Q(W)(T,T,T,T)\leq \tau_{-}^{2}Q(W)(\overline{K},T,T,T).
\end{equation}

\subsection{Some quantitative results from \cite{Christodoulou-Klainerman}}
The work of \cite{Christodoulou-Klainerman} not only proves the global nonlinear stability of Minkowski space but also gives precise asymptotics for various geometric quantities. We refer to Chapter 10 of \cite{Christodoulou-Klainerman} for the results in this subsection.

Let $C_{0}$ be the future complete null cone with vertex at a fixed point on $\Sigma_{-1}$. For $t\geq 0$, we use $r_{0}(t)$ to represent the radius of $\Sigma_{t}\cap C_0$. The function $r_{0}(t)$ and $t$ are comparable in the following sense:
\begin{equation}\label{est:r0}
	\frac{1}{2}(1+t)\leq r_{0}(t)\leq\frac{3}{2}(1+t).
\end{equation}

The exterior region $\Sigma_t^e$ in $\Sigma_{t}$ is defined to be the region $\big\{p\in \Sigma_{t}\big|r(p)\geq \frac{1}{2}r_{0}(t)\big\}$ and the interior region $\Sigma_i^e$ in $\Sigma_{t}$ is defined to be the region $\big\{p\in \Sigma_{t}\big|r(p)\leq \frac{1}{2}r_{0}(t) \big\}$.  In \cite{Christodoulou-Klainerman}, one constructs a global smooth exterior optical function $u$, namely a solution of the Eikonal equation defined everywhere in the exterior region $r\geq\frac{r_{0}}{2}$.

For any $S_{t,u}$-tangential tensor fields $V$, we introduce the following norms defined in the interior region $\Sigma_{t}^{i}$ and exterior region $\Sigma_{t}^{e}$:
\[\|V\|_{p,i}=\|V\|_{L^{p}(\Sigma_t^i)}, \ \ \|V\|_{p,e}=\|V\|_{L^{p}(\Sigma_t^e)}.\]
where $p=2$ or $p=\infty$.

The lapse function $\phi$ satisfies the following estimates:
\begin{equation}\label{est:lapse}
	r_{0}\|\phi-1 \|_{\infty,i}\lesssim\varepsilon_{0}, \ \ \|r(\phi-1)\|_{\infty,e}\lesssim\varepsilon_{0} , \ \ r_{0}^{-\frac{1}{2}}||r^{\frac{5}{2}}\slashed{\nabla}\phi||_{\infty,e}\lesssim\varepsilon_{0}, \ \ r_{0}^{\frac{1}{2}}\|r^{\frac{3}{2}}\slashed{\nabla}_{N}\phi\|_{\infty,e}\lesssim\varepsilon_{0}, \ \ r_{0}^{2}\|D\phi\|_{\infty,i}\lesssim\varepsilon_{0} .
\end{equation}

The curvature tensor satisfies the following interior estimates
\begin{equation} \label{est: curvature int}
	r_0^{\frac{7}{2}}\|\left(\alpha,  \underline{\alpha},  \beta, \underline{\beta}, \rho, \sigma\right)\|_{\infty,i}\lesssim\varepsilon_0,
\end{equation}
and exterior estimates
\begin{equation}\label{est: curvature}
	\|\left(r^{\frac{7}{2}}\alpha, r\tau_{-}^{\frac{5}{2}}\underline{\alpha}, r^{\frac{7}{2}}\beta, r^{2}\tau_{-}^{\frac{3}{2}}\underline{\beta},r^{3}\rho,r^{3}\sigma\right)\|_{\infty,e}+\|r^{\frac{9}{2}}\slashed{\nabla}\beta\|_{\infty,e} \lesssim\varepsilon_0.
\end{equation}

The second fundamental form $k$ satisfies the following interior estimates:
\begin{equation}\label{est: second fundamental form interior}
	r_{0}^{2}\|k\|_{\infty,i}+r_{0}^{3}\|\slashed{\nabla}k\|_{\infty,i}\lesssim\epsilon_{0}
\end{equation}
and exterior estimates
\begin{equation}\label{est: second fundamental form exterior}
	r_{0}^{-\frac{1}{2}}\|\left(r^{\frac{5}{2}}\delta,r^{\frac{5}{2}}\epsilon,\min\{\tau_{-}^{\frac{3}{2}}r_{0}^{\frac{1}{2}},r^{\frac{3}{2}}\}r\widehat{\Xi} \right)\|_{\infty,e}+r_{0}^{-\frac{1}{2}}\|\left(r^{\frac{7}{2}}\slashed{\nabla}\delta,r^{\frac{7}{2}}\slashed{\nabla}\epsilon\right)\|_{\infty,e}\lesssim\epsilon_{0}
\end{equation}

For the connection coefficients, we have the following interior estimates:
\begin{equation}\label{est: Hessian int}
	r_{0}^{\frac{3}{2}}\|\left({\rm tr}\chi-\overline{{\rm tr}\chi},(\overline{{\rm tr}\chi}-\frac{2}{r} ),\widehat{\chi},\eta\right)\|_{i,\infty}\lesssim\varepsilon_0
\end{equation}
and exterior estimates
\begin{equation}\label{est: Hessian ext}
	\|\left( r^{2}({\rm tr}\underline{\chi}+\frac{2}{r}),r^{2}({\rm tr}\chi-\frac{2}{r}),r^{2}\widehat{\chi},r^{2}\widehat{\underline{\chi}},r^{2}\eta,r^{2}\underline{\eta},r^{2}\omega,r^{2}\underline{\omega},r^{2}\zeta \right)\|_{e,\infty}\lesssim\varepsilon_0
\end{equation}
For $a=\frac{1}{|\nabla u|}$,we have the following estimates:
\begin{equation}\label{est: a}
	\|r(a-1)\|_{e,\infty}\lesssim\varepsilon_0, \ r_{0}^{\frac{1}{2}}\|a-1\|_{i,\infty}\lesssim\varepsilon_{0}
\end{equation}

In view of \eqref{deformation: T} and \eqref{deformation: K},by a universal constant, the above estimates imply that  the trace free parts of the deformation tensor of $T$ and $\overline{K}$ satisfy the following estimates
\begin{equation}\label{est: deformation tensor T Kb 1}
	r_{0}^{\frac{3}{2}}\|{}^{(T)}\widehat{\pi}\|_{\infty,i}+r_{0}^{-\frac{1}{2}}\|{}^{(\overline{K})}\widehat{\pi}\|_{\infty,i}\lesssim\varepsilon_0,
\end{equation}
\begin{equation}\label{est: deformation tensor T Kb 2}
		\|r\tau_{-}{}^{(T)}\widehat{\pi}_{AB}\|_{\infty,e}+\|r^{2}\left({}^{(T)}\widehat{\pi}_{34},{}^{(T)}\widehat{\pi}_{4A},{}^{(T)}\widehat{\pi}_{3A},{}^{(T)}\widehat{\pi}_{44},{}^{(T)}\widehat{\pi}_{33}\right)\|_{\infty,e}\lesssim\varepsilon_0,
\end{equation}		
\begin{equation}\label{est: deformation tensor T Kb 3}
		\|\left({}^{(\overline{K})}\widehat{\pi}_{AB},{}^{(\overline{K})}\widehat{\pi}_{34},{}^{(\overline{K})}\widehat{\pi}_{33},{}^{(\overline{K})}\widehat{\pi}_{3A}\right)\|_{\infty,e}+
		\|r^{2}\tau_{-}^{-2}({}^{(\overline{K})}\widehat{\pi}_{4A},{}^{(\overline{K})}\widehat{\pi}_{44})\|_{\infty,e}\lesssim\varepsilon_0.
\end{equation}

\section{The main theorem}
We consider the null Bianchi equation \eqref{eq: null bianchi: 4 alphab} for $\underline{\alpha}$ along the null cone $C_{u}$. In Chapter 17 of \cite{Christodoulou-Klainerman}, one shows that \eqref{eq: null bianchi: 4 alphab} implies the radiation field $\underline{\mathbf{A}}$ can be defined as follows on the future null infintiy:
\begin{equation}\label{def:Ab}
\underline{\mathbf{A}}:=\lim_{t\rightarrow 0, \atop \\ {\text on}~C_u}r\underline{\alpha},
\end{equation}
where the limit is taken along each null geodesic along a fixed light cone $C_u$. The radiation field $\underline{\mathbf{A}}$ can be interpreted as the gravitation waves detected from a far-away observer (at future null infinity). The main theorem of the paper will verify the following physical intuition: if no gravitational waves are detected by the far-away observers, then there are no gravitational waves at all, i.e., the spacetime is flat hence isometric to the Minkowski spacetime. We refer to \cite{Li-Yu} for a similar result for the Alfv\'en waves.
\begin{Theorem}
For sufficiently small $\varepsilon_0$, if the $\underline{\mathbf{A}}$ vanishes, then $M$ is isometric to the Minkowski spacetime.
\end{Theorem} 
\begin{Remark}
 The theorem is analogous to the rigidity part of the positive mass theorem proved by Schoen and Yau in \cite{Schoen-Yau1} and  \cite{Schoen-Yau2}. In that case, if the ADM mass $m$, which is defined at spatial infinity, vanishes, one can also conclude that the spacetime is flat.
 
Indeed, we can also use the positive mass theorem to prove the main theorem of the paper, see Remark \ref{remark: using positive mass theorem} at the end of the paper.
\end{Remark}
In the current work, since we only deal with the case that $\underline{\mathbf{A}}\equiv 0$, it suffices to consider the limit of $|r\underline{\alpha}|$ at future null infinity. We can take the inner product with $r^2\underline{\alpha}$ for both sides of  \eqref{eq: null bianchi: 4 alphab} to derive a transport equation for $r^2|\underline{\alpha}|^2$:
\begin{equation}\label{eq: nabla 4 of r2alphab2}
\begin{split}
 e_{4}(r^{2}|\underline{\alpha}|^{2})&=ar^{2}|\underline{\alpha}|^{2}\left(\overline{a^{-1}{\rm tr}\chi}-a^{-1}{\rm tr}\chi\right)-2r^{2}\underline{\alpha}^{AB} \left(\slashed{\nabla}\widehat{\otimes}\underline{\beta}\right)_{AB}+8r^{2}\omega|\underline{\alpha}|^{2}\\
 &\ \ \ -6r^{2}\underline{\alpha}^{AB}(\widehat{\underline{\chi}}_{AB}\rho-\,^*\widehat{\underline{\chi}}_{AB}\sigma)+2r^{2}\underline{\alpha}^{AB}\left[(\zeta-4\eta)\widehat{\otimes}\underline{\beta}\right]_{AB}.
 \end{split}
 \end{equation}
 We recall that $L=-Du$ is related to $e_4$ by $L=a^{-1}e_4$. For a given null geodesic generated by $e_4$ or $L$, we take $s$ to be the affine parameter for $L$. Hence,
 \begin{align*}
 	\frac{\partial r}{\partial s}=\frac{r}{2}\overline{a^{-1}{\rm tr}\chi}.
 \end{align*}
 Therefore, along the null geodesic, we can rewrite \eqref{eq: nabla 4 of r2alphab2} as
\begin{equation}\label{eq: L of r2alphab2}
\begin{split}
 \frac{d}{ds}(r^{2}|\underline{\alpha}|^{2})&=r^{2}|\underline{\alpha}|^{2}\left(\overline{a^{-1}{\rm tr}\chi}-a^{-1}{\rm tr}\chi\right)-2a^{-1}r^{2}\underline{\alpha}^{AB} \left(\slashed{\nabla}\widehat{\otimes}\underline{\beta}\right)_{AB}+8a^{-1}r^{2}\omega|\underline{\alpha}|^{2}\\
 &\ \ \ -6a^{-1}r^{2}\underline{\alpha}^{AB}(\widehat{\underline{\chi}}_{AB}\rho-\,^*\widehat{\underline{\chi}}_{AB}\sigma)+2a^{-1}r^{2}\underline{\alpha}^{AB}\left[(\xi-4\eta)\widehat{\otimes}\underline{\beta}\right]_{AB}.
 \end{split}
 \end{equation}
We can apply the bounds \eqref{est: curvature} and \eqref{est: Hessian ext} to the righthand side. This leads to 
\[ \frac{d}{ds}(r^{2}|\underline{\alpha}|^{2})\lesssim \frac{\varepsilon_0}{(1+r)^{\frac{3}{2}}} r |\underline{\alpha}|.\]
Hence,
\[\frac{d}{ds}(r|\underline{\alpha}|)\lesssim \frac{\varepsilon_0}{(1+r)^{\frac{3}{2}}}.\]
The right hand side is integrable for $r\in [0,\infty)$.
Therefore, for any point $p\in M$, we can integrate $\frac{d}{ds}(r|\underline{\alpha}|)$ along the outgoing null geodesic emanating from the point $p$. Therefore, we have
\begin{equation}\label{eq: vanishing condition}
	0=|r\underline{\alpha}|(p)+\int_{0}^{\infty} \frac{d}{d s}\left(|r\underline{\alpha}|\right)ds.
\end{equation}

\section{Proof of the main theorem}
We start with an estimate on the volume of annular regions near light cone.
\begin{Lemma}\label{cor:volume}
	For any $t\geq 0$ and $r_{1}, r_{2} \in \mathbb{R}$ with $\frac{1}{2}r_{0}(t)\leq r_{1}< r_{2}$, for sufficiently small $\varepsilon_0$, the volume of annulus $\big\{p\in \Sigma_{t}\big| r_{1}\leq r(p)\leq r_{2} \big\}$ is bounded above by $C(r_{2}^{3}-r_{1}^{3})$ where $C$ is a universal constant.\end{Lemma}
\begin{proof}
The annular region $\big\{p\in \Sigma_{t}\big| r_{1}\leq r(p)\leq r_{2} \big\}$ is diffeomorphic to $ [u_{1},u_{2}]\times \mathbf{S}^2$. In terms of local coordinates $(u,\vartheta_1,\vartheta_2)\in [u_{1},u_{2}]\times \mathbf{S}^2$, the metric $g$ restricted to the annular region takes the following form
\[a^{2}du^{2}+\sum_{A,B=1,2}\slashed{g}_{AB}d\vartheta^{A}d\vartheta^{B}.\]
The volume of the region is given by
\[\int_{u_{1}}^{u_{2}}\left(\int_{\mathbf{S}^{2}} a\sqrt{\det\slashed{g}}d\vartheta^{A}d\vartheta^{B}\right).\]
In view of \eqref{est: Hessian ext} and \eqref{est: a}, it is bounded by
\[\int_{u_{1}}^{u_{2}} {\rm Area}(S_{t,u})du=\int_{r_{1}}^{r_{2}} 4\pi r^{2}\frac{r}{2}\overline{atr\theta}dr\lesssim\int_{r_{1}}^{r_{2}}r^{2}dr=\frac{1}{3}(r_{2}^{3}-r_{1}^{3}) \]
where $u_{i}$ is chosen in such a way that ${\rm Area}(S_{t,u_{i}})=4\pi r_{1}^{2}$, $i=1,2$.
\end{proof}

\subsection{Improved pointwise decay on curvature}
\begin{Lemma}
Under the assumption of main theorem, the decay estimates for $\underline{\alpha}$ can be improved to
\[\tau_{+}^{\frac{3}{2}}\tau_{-}^{2}|\underline{\alpha}|\lesssim \varepsilon_0.\]
\end{Lemma}
\begin{proof}
The key idea is to integrate from null infinity. We notice that in the region where $\{r\leq r_{0} \}\cup\{r\geq 2r_{0}\}$, the decay for $\underline{\alpha}$ is already satisfactory:
 \[|\underline{\alpha}_{AB}|\lesssim\varepsilon_0\tau_{+}^{-1}\tau_{-}^{-\frac{5}{2}}\lesssim\varepsilon_0\tau_{+}^{-\frac{3}{2}}\tau_{-}^{-2}.\]

It remains to prove the improved decay in the region where $\{\frac{1}{2}r_{0}\leq r\leq 2r_{0} \}$. Since $e_4(\tau_-)=0$, we can multiply \eqref{eq: nabla 4 of r2alphab2} by the weight $\tau_-^4$ to derive
\begin{align*}
 \frac{d}{ds}(r^{2}\tau_-^4|\underline{\alpha}|^{2})=&r^{2}\tau_-^4|\underline{\alpha}|^{2}\left(\overline{a^{-1}{\rm tr}\chi}-a^{-1}{\rm tr}\chi\right)-2a^{-1}r^{2}\tau_-^2\underline{\alpha}^{AB} \left(\nabla\widehat{\otimes}(\tau_-^2\underline{\beta})\right)_{AB}\\
 &+8a^{-1}\tau_-^4r^{2}\omega|\underline{\alpha}|^{2} -6a^{-1}r^{2}\tau_-^2\underline{\alpha}^{AB}(\widehat{\underline{\chi}}_{AB}(\tau_-^2\rho)-\,^*\widehat{\underline{\chi}}_{AB}(\tau_-^2\sigma))\\
 &+2a^{-1}r^{2}\tau_-^2\underline{\alpha}^{AB}\left[(\zeta-4\eta)\widehat{\otimes}(\tau_-^2\underline{\beta})\right]_{AB}.
 \end{align*}

We can apply the bounds \eqref{est: curvature},\eqref{est: Hessian ext} and \eqref{est: a} to the righthand side. This leads to 
\[ \frac{d}{ds}\left(r^{2}\tau_-^4|\underline{\alpha}|^{2}\right)\lesssim \frac{\varepsilon_0}{(1+r)^{\frac{3}{2}}} r\tau_-^2 |\underline{\alpha}|.\]
Hence,
\[\frac{d}{ds}\left(r\tau_-^2|\underline{\alpha}|\right)\lesssim \frac{\varepsilon_0}{(1+r)^{\frac{3}{2}}}.\]
We can integrate it from $\infty$ to a fixed $s=t$ where $r$ corresponds to $r_0$. Thus, \eqref{eq: vanishing condition} implies that
\[	|r_{0}\tau_-^2\underline{\alpha}(t,\cdot)|\lesssim \varepsilon_0\int_{t}^{\infty} r_{0}^{-\frac{3}{2}}ds\lesssim \varepsilon_0 r^{-\frac{1}{2}}.\]
This proves the lemma.
\end{proof}
\subsection{Improved pointwise decay on the second fundamental form}
We recall that $\Xi$ is the projection of the second fundamental form $k$ to $S_{t,u}$, i.e., $\Xi_{AB}=k_{AB}$.

\begin{Lemma}
Under the assumption of main theorem,$\widehat{\Xi}$ satisfies the following decay estimate in the region $\{\frac{r_{0}}{2}\leq r\leq 2r_{0} \}$.
\[\tau_{+}^{\frac{3}{2}}|\widehat{\Xi}|\lesssim\varepsilon_0.\]

\end{Lemma}
\begin{proof}
On $\Sigma_t$, we have the following equations for $\widehat{\Xi}$:
\begin{align}
	\slashed{\nabla}_{N}\widehat{\Xi}+\frac{1}{2}{\rm tr}\theta\widehat{\Xi}=\frac{1}{4}(-\underline{\alpha}+\alpha)+\frac{1}{2}\slashed{\nabla}\widehat{\otimes}\zeta+\frac{3}{2}\delta\widehat{\theta}+(a^{-1}\slashed{\nabla} a)\widehat{\otimes}\zeta.
\end{align}
In view of the fact that $\slashed{\nabla}_{N}r=\frac{r}{2a}\overline{a{\rm tr}\theta}$, we obtain
\begin{align*}
	a\nabla_{N}\left(r^2|\widehat{\Xi}|^2\right)&=r^{2}|\widehat{\Xi}|^{2}(\overline{a {\rm tr}\theta}-a {\rm tr}\theta)+\frac{a}{2}r^{2}\widehat{\Xi}^{AB}(-\underline{\alpha}_{AB}+\alpha_{AB})+ar^{2}\widehat{\Xi}^{AB}(\slashed{\nabla}\widehat{\otimes}\zeta)_{AB}\\
	&\ \ \ +3ar^{2}\delta\widehat{\Xi}^{AB}\widehat{\theta}_{AB}+2r^{2}\widehat{\Xi}^{AB}(\slashed{\nabla}\widehat{\otimes}\zeta)_{AB}.
\end{align*}
By virtue of \eqref{est: curvature}, \eqref{est: second fundamental form exterior}, \eqref{est: Hessian ext} and \eqref{est: a},we have
\begin{align*}
	\left|\nabla_{N}(r|\widehat{\Xi}|)^{2}+\frac{1}{2}(r^{2}\underline{\alpha}_{AB}\widehat{\Xi}^{AB})\right|\lesssim \varepsilon_0 r^{-1}\tau_{-}^{-\frac{3}{2}}|r\widehat{\Xi}|.
\end{align*}
According to the improved decay on $\underline{\alpha}$, we have
\[	|r^{2}\underline{\alpha}_{AB}\widehat{\Xi}^{AB}|\lesssim\varepsilon_0\tau_{+}^{-\frac{1}{2}}\tau_{-}^{-2}|r\widehat{\Xi}|.
\]
Therefore, 
\[\left|\nabla_{N}|r\widehat{\Xi}|\right|\lesssim\varepsilon_0\tau_{+}^{-\frac{1}{2}}\tau_{-}^{-\frac{3}{2}}.\]
Since $\frac{1}{2}r_{0}\leq r\leq 2r_{0}$, we integrate the above equation from spatial infinity to derive
\begin{align*}
	\left|r\widehat{\Xi}\right|=|\int_{r}^{\infty}a\nabla_{N}(r'|\widehat{\Xi}|)dr'\lesssim\varepsilon_0\tau_{+}^{-\frac{1}{2}}.
\end{align*}
This gives the desired bound.
\end{proof}

\begin{Corollary}
Under the assumption of main theorem, the decay estimates for ${}^{(T)}\widehat{\pi}$ in the exterior region can be improved to 
\[	\|\tau_{+}^{\frac{3}{2}}\,{}^{(T)}\widehat{\pi}_{AB}\|_{\infty,e}\lesssim\varepsilon_0.\]
\end{Corollary}
This is clear from the formula that ${}^{(T)}\pi_{AB}=-2\Xi_{AB}$.

\subsection{Decay in the energy estimates}
For each $t\geq 0$, we define
 \[E(t)=\int_{\Sigma_{t}}Q(T,T,T,T).\]
 We recall that $Q$ is the Bel-Robinson tensor associated to the curvature tensor of the spacetime. In view of \eqref{eq: computation for Q}, to show that the curvature is vanishing, it suffices to show that $E(t)\equiv 0$. \begin{Lemma}\label{lemma: lemma5}
Under the assumption of main theorem, for all $t_1, t_2\geq 0$, $t_1<t_2$, we have
\[
	|E(t_{2})-E(t_{1})|\lesssim \varepsilon_0 \sup_{s\in [t_{1},t_{2}]}E(s).
\]
\end{Lemma}
\begin{proof}
We apply the basic energy identity  \eqref{basic energy estimates} for $(X,Y,Z)=(T,T,T)$ in the spacetime region bounded by $\Sigma_{t_1}$ and $\Sigma_{t_2}$. This leads to
\begin{equation}\label{eq:a1}
E(t_{2})-E(t_{1})=\frac{3}{2}\int_{t_{1}}^{t_{2}}\phi \left(\int_{\Sigma_{s}}Q_{\alpha\beta\gamma\delta}{}^{(T)}\widehat{\pi}^{\alpha\beta}T^{\gamma}T^{\delta}\right)ds.
\end{equation}
We compute the integrand $Q_{\alpha\beta\gamma\delta}{}^{(T)}\widehat{\pi}^{\alpha\beta}T^{\gamma}T^{\delta}$ as follows:
\begin{align*}
Q_{\alpha\beta\gamma\delta}{}^{(T)}\widehat{\pi}^{\alpha\beta}T^{\gamma}T^{\delta}&=\frac{1}{4}Q(T,T,e_4,e_4){}^{(T)}\widehat{\pi}_{33}+\frac{1}{4}Q(T,T,e_3,e_3){}^{(T)}\widehat{\pi}_{44}+\frac{1}{2}Q(T,T,e_3,e_4){}^{(T)}\widehat{\pi}_{43}\\
&\ \ -\frac{1}{2}Q(T,T,e_4,e_A){}^{(T)}\widehat{\pi}_{3A}-\frac{1}{2}Q(T,T,e_3,e_A){}^{(T)}\widehat{\pi}_{4A}+\frac{1}{4}Q(T,T,e_A,e_B){}^{(T)}\widehat{\pi}_{AB}.
\end{align*}
By writing $T$ as $\frac{1}{2}(e_4+e_3)$, we can bound each $Q(\cdot,\cdot,\cdot,\cdot)$ term in the above formula by a constant times $Q(T,T,T,T)$.  To bound the right hand side of \eqref{eq:a1}, it suffices to bound the $\widehat{\pi}$-terms.

In the interior region, in view of \eqref{est: deformation tensor T Kb 1}, we have
\begin{align*}
	\left|\int_{\Sigma^i_{s}}Q_{\alpha\beta\gamma\delta}{}^{(T)}\widehat{\pi}^{\alpha\beta}T^{\gamma}T^{\delta}\right|\lesssim \varepsilon_0 r_{0}^{-\frac{3}{2}}\int_{\Sigma_{s}}Q(T,T,T,T)= \varepsilon_0 r_{0}^{-\frac{3}{2}}E(s).
\end{align*}
In the exterior region, in view of \eqref{est: deformation tensor T Kb 2} and the improved estimates of ${}^{(T)}\widehat{\pi}$ in the exterior region, we have
\begin{align*}
	\left|\int_{\Sigma^e_{s}}Q_{\alpha\beta\gamma\delta}{}^{(T)}\widehat{\pi}^{\alpha\beta}T^{\gamma}T^{\delta}\right|\lesssim \varepsilon_0 r_{0}^{-\frac{3}{2}}\int_{\Sigma_{s}}Q(T,T,T,T)= \varepsilon_0 r_{0}^{-\frac{3}{2}}E(s).
\end{align*}
Putting these two estimates together, we have
\begin{align*}
	\left|\int_{\Sigma_{s}}Q_{\alpha\beta\gamma\delta}{}^{(T)}\widehat{\pi}^{\alpha\beta}T^{\gamma}T^{\delta}\right|\lesssim  \varepsilon_0 r_{0}^{-\frac{3}{2}}E(s).
\end{align*}
In view of \eqref{est:r0} and \eqref{est:lapse}, the equation \eqref{eq:a1} leads to
\begin{align*}
\left|E(t_{2})-E(t_{1})\right|&\lesssim \varepsilon_0\int_{t_{1}}^{t_{2}}r_{0}(s)^{-\frac{3}{2}}E(s)ds\\
&\lesssim \varepsilon_0\left(\int_{t_{1}}^{t_{2}}(1+s)^{-\frac{3}{2}}ds\right)\sup_{s\in[t_{1},t_{2}]}E(s)\lesssim \varepsilon_0\sup_{s\in[t_{1},t_{2}]}E(s).
\end{align*}
This completes proof.
\end{proof}

\begin{Lemma}
Under the assumption of main theorem, we have the following limit:
\[\lim_{t\rightarrow\infty}E(t)=0.\]
\end{Lemma}
\begin{proof}We apply the basic energy identity  \eqref{basic energy estimates} for $(X,Y,Z)=(\overline{K},T,T)$ in the spacetime region bounded by $\Sigma_{t_1}$ and $\Sigma_{t_2}$. This leads to
\begin{align*}
	&\int_{\Sigma_{t_{2}}}Q (\overline{K},T,T,T) -\int_{\Sigma_{t_{1}}}Q(\overline{K},T,T,T)\\
=&\int_{t_{1}}^{t_{2}}\phi \left(\int_{\Sigma_{s}}Q_{\alpha\beta\gamma\delta}{}^{(T)}\widehat{\pi}^{\alpha\beta}\overline{K}^{\gamma}T^{\delta}\right)ds+\frac{1}{2}\int_{t_{1}}^{t_{2}}\phi \left(\int_{\Sigma_{s}}Q_{\alpha\beta\gamma\delta}{}^{(\overline{K})}\widehat{\pi}^{\alpha\beta}T^{\gamma}T^{\delta}\right)ds
\end{align*}
We compute the integrands $Q_{\alpha\beta\gamma\delta}{}^{(T)}\widehat{\pi}^{\alpha\beta}\overline{K}^{\gamma}T^{\delta}$ and $Q_{\alpha\beta\gamma\delta}{}^{(\overline{K})}\widehat{\pi}^{\alpha\beta}T^{\gamma}T^{\delta}$ on the righthand side of the above equation as follows:
\begin{equation}\label{eq:Q1}
\begin{split}
Q_{\alpha\beta\gamma\delta}{}^{(T)}\widehat{\pi}^{\alpha\beta}T^{\gamma}T^{\delta}&=\underbrace{\frac{1}{4}Q(\overline{K},T,e_4,e_4){}^{(T)}\widehat{\pi}_{33}}_{I_1}+\underbrace{\frac{1}{4}Q(\overline{K},T,e_3,e_3){}^{(T)}\widehat{\pi}_{44}}_{I_2}\\
&\ \ +\underbrace{\frac{1}{2}Q(\overline{K},T,e_3,e_4){}^{(T)}\widehat{\pi}_{43}}_{I_3}\underbrace{-\frac{1}{2}Q(\overline{K},T,e_4,e_A){}^{(T)}\widehat{\pi}_{3A}}_{I_4}\\
&\ \ \underbrace{-\frac{1}{2}Q(\overline{K},T,e_3,e_A){}^{(T)}\widehat{\pi}_{4A}}_{I_5}+\underbrace{\frac{1}{4}Q(\overline{K},T,e_A,e_B){}^{(T)}\widehat{\pi}_{AB}}_{I_6},
\end{split}
\end{equation}
and
\begin{equation}\label{eq:Q2}
\begin{split}
Q_{\alpha\beta\gamma\delta}{}^{(\overline{K})}\widehat{\pi}^{\alpha\beta}T^{\gamma}T^{\delta}&=\frac{1}{4}Q(T,T,e_4,e_4){}^{(\overline{K})}\widehat{\pi}_{33}+\frac{1}{4}Q(T,T,e_3,e_3){}^{(\overline{K})}\widehat{\pi}_{44}\\
& \ \ +\frac{1}{2}Q(T,T,e_3,e_4){}^{(\overline{K})}\widehat{\pi}_{43} -\frac{1}{2}Q(T,T,e_4,e_A){}^{(\overline{K})}\widehat{\pi}_{3A}\\
&\ \ -\frac{1}{2}Q(T,T,e_3,e_A){}^{(\overline{K})}\widehat{\pi}_{4A}+\frac{1}{4}Q(T,T,e_A,e_B){}^{(\overline{K})}\widehat{\pi}_{AB}.
\end{split}
\end{equation}
We recall that
 \[\overline{K}=T+ K=\frac{1}{2}\left(\tau_-^2e_{3}+\tau_+^2e_{4}\right).\]
In the first case where $Q$ is contracted with $\overline{K}$, we bound $Q$ by $\tau_{+}^{2} Q(T,T,T,T)$; in the second case, we bound $Q$ by $Q(T,T,T,T)$.

In the interior region, in view of \eqref{est: curvature}, \eqref{est: deformation tensor T Kb 1}, \eqref{est: deformation tensor T Kb 2} and \eqref{est: deformation tensor T Kb 3},   we have
\begin{align*}
\left|\int_{t_{1}}^{t_{2}}\phi \left(\int_{\Sigma^i_{s}}Q_{\alpha\beta\gamma\delta}{}^{(T)}\widehat{\pi}^{\alpha\beta}\overline{K}^{\gamma}T^{\delta}\right)ds\right|&\lesssim \int_{t_{1}}^{t_{2}}\left(\int_{\Sigma^i_{s}}\tau_{+}^{2}\left| Q(T,T,T,T)\right|\cdot \left|{}^{(T)}\widehat{\pi}\right|\right)ds\\
&\lesssim \int_{t_{1}}^{t_{2}}\left(\int_{\Sigma^i_{s}}\tau_{+}^{2} r_0(s)^{-6}\varepsilon_0^2\cdot r_0(s)^{-\frac{3}{2}}\varepsilon_0\right)ds.
\end{align*}
Since $\tau_{+} \sim r_0(s) \sim s$ in the interior region, we derive that
\begin{align*}
\left|\int_{t_{1}}^{t_{2}}\phi \left(\int_{\Sigma^i_{s}}Q_{\alpha\beta\gamma\delta}{}^{(T)}\widehat{\pi}^{\alpha\beta}\overline{K}^{\gamma}T^{\delta}\right)ds\right|& \lesssim \varepsilon_0^3\int_{t_{1}}^{t_{2}}s^{-\frac{5}{2}}ds\lesssim \varepsilon_0^3 
\end{align*}
Similarly, we have
\begin{align*}
\left|\int_{t_{1}}^{t_{2}}\phi \left(\int_{\Sigma^i_{s}}Q_{\alpha\beta\gamma\delta}{}^{(\overline{K})}\widehat{\pi}^{\alpha\beta}T^{\gamma}T^{\delta}\right)ds\right|& \lesssim \varepsilon_0^3.
\end{align*}
In fact, for the region where $r\geq 2r_0$, since $\tau_-\sim \tau_+$, all the components in the above prove enjoy the same estimates as in the interior region. We can then proceed exactly in the same way to prove that the corresponding error integrals are bounded above by a universal constant times $\varepsilon_0^3$. We omit the details since they are straightforward.

It remains to bound the two error integrals where $r\in [\frac{r_0}{2},2r_0]$. According to the decomposition in \eqref{eq:Q1}, we decompose the error integral as follows:
\begin{align*}
\left|\int_{t_{1}}^{t_{2}}\phi \left(\int_{ \{\frac{r_0}{2}\leq r \leq 2r_0\}}Q_{\alpha\beta\gamma\delta}{}^{(T)}\widehat{\pi}^{\alpha\beta}\overline{K}^{\gamma}T^{\delta}\right)ds\right|&=\mathbf{I}_1+\mathbf{I}_2+\mathbf{I}_3+\mathbf{I}_4+\mathbf{I}_5+\mathbf{I}_6.
\end{align*}
From the decay point of view, the most difficult terms are $\mathbf{I}_6$ and $\mathbf{I}_2$ . We first compute that
\begin{align*}
Q_{44AB}&=2|\beta|^2\delta_{AB}+2\rho \alpha_{AB}-2\sigma \,^{*}\alpha_{AB},\\
Q_{33AB}&=2|\beta|^2\delta_{AB}+2\rho \underline{\alpha}_{AB}+2\sigma \,^{*}\underline{\alpha}_{AB},\\
Q_{34AB}&=2(\beta\cdot \underline{\beta}+|\rho|^2+|\sigma|^2)\delta_{AB}-2(\beta_A\underline{\beta}_B+\beta_B\underline{\beta}_A).
\end{align*}
Therefore, we have
\begin{align*}
I_6&=\frac{1}{16}\left(\tau_-^2Q_{33AB}+\tau_+^2Q_{44AB}+(\tau_-^2+\tau_+^2)Q_{34AB}\right){}^{(T)}\widehat{\pi}_{AB}\\
&=\left[\frac{\tau_-^2}{8}\left(\rho \underline{\alpha}_{AB}+\sigma \,^{*}\underline{\alpha}_{AB}\right)+\frac{\tau_+^2}{8}\left(\rho{\alpha}_{AB}-\sigma \,^{*}{\alpha}_{AB}\right) -\frac{\tau_-^2+\tau_+^2}{4}\beta_A\underline{\beta}_B\right]{}^{(T)}\widehat{\pi}_{AB}.
\end{align*}
In view of \eqref{est: curvature}, in particular the $\tau_-$-weight in $\underline{\alpha}$,  to obtain that
\[|\frac{\tau_-^2}{8}\left(\rho \underline{\alpha}_{AB}+\sigma \,^{*}\underline{\alpha}_{AB}\right)|\lesssim \tau_+^{-4}\varepsilon_0^2.\]
Similarly, we have
\[|\frac{\tau_+^2}{8}\left(\rho{\alpha}_{AB}-\sigma \,^{*}{\alpha}_{AB}\right) |\lesssim\tau_+^{-\frac{9}{2}}\varepsilon_0^2,\]
and 
\[|\frac{\tau_-^2+\tau_+^2}{4}\beta_A\underline{\beta}_B|\lesssim \tau_+^{-\frac{7}{2}}\varepsilon_0^2.\]
Therefore,
\[|I_6|\lesssim \tau_+^{-5}\varepsilon_0^3.\]
Therefore, we have
\[\mathbf{I}_6\lesssim \int_{t_1}^{t_2} \tau_+^{-2}\varepsilon_0^3\lesssim \varepsilon_0^3.\]
For $\mathbf{I}_2$, using the fomulas
\[Q_{3333}=2|\underline{\alpha}|^2, Q_{4433}=4\left(|\rho|^2+|\sigma|^2\right), Q_{4333}=4|\underline{\beta}|^2,\]
we first compute
\begin{align*}
I_2&=\frac{1}{16}\left(\tau_-^2Q_{3333}+\tau_+^2Q_{4433}+(\tau_-^2+\tau_+^2)Q_{3433}\right){}^{(T)}\widehat{\pi}_{44}\\
&=\left(\frac{1}{8}\tau_-^2|\underline{\alpha}|^2+\frac{\tau_+^2}{4}\left(|\rho|^2+|\sigma|^2\right)+\frac{\tau_-^2+\tau_+^2}{4}|\underline{\beta}|^2\right){}^{(T)}\widehat{\pi}_{44}
\end{align*}
In view of \eqref{est: curvature} and improved estimate on $\underline{\alpha}$, we obtain that
\[\frac{1}{8}\tau_-^2|\underline{\alpha}|^2+\frac{\tau_+^2}{4}\left(|\rho|^2+|\sigma|^2\right)+\frac{\tau_-^2+\tau_+^2}{4}|\underline{\beta}|^2\lesssim \varepsilon_0^2 \tau_+^{-2}.\]
Therefore,
\[|I_2|\lesssim \tau_+^{-4}\varepsilon_0^3.\]
Hence,
Therefore, we have
\[\mathbf{I}_2\lesssim \int_{t_1}^{t_2} \tau_+^{-1}\varepsilon_0^3\lesssim \varepsilon_0^3 \log(t_2).\]
The rest of the ${I}_i$'s can be computed as follows: 
\begin{align*}
I_1&=\frac{1}{16}\left(\tau_-^2Q_{3344}+(\tau_-^2+\tau_+^2)Q_{3444}+\tau_+^2Q_{4444}\right){}^{(T)}\widehat{\pi}_{33}\\
&=\frac{1}{8}\left(2\tau_-^2(\rho^2+\sigma^2)+2(\tau_-^2+\tau_+^2)|\beta|^2+\tau_+^2|\alpha|^2\right){}^{(T)}\widehat{\pi}_{33},
\end{align*}
\begin{align*}
I_3&=\frac{1}{8}\left(\tau_-^2Q_{3334}+(\tau_-^2+\tau_+^2)Q_{3344}+\tau_+^2Q_{3444}\right){}^{(T)}\widehat{\pi}_{34}\\
&=\frac{1}{2}\left(\tau_-^2|\underline{\beta}|^2+(\tau_-^2+\tau_+^2)(\rho^2+\sigma^2)+\tau_+^2|\beta|^2\right){}^{(T)}\widehat{\pi}_{34},
\end{align*}
\begin{align*}
I_4&=-\frac{1}{8}\left(\tau_-^2Q_{334A}+\tau_-^2Q_{344A}+\tau_+^2Q_{434A}+\tau_+^2Q_{444A}\right){}^{(T)}\widehat{\pi}_{3A}\\
&=-\frac{1}{2}\left[\tau_-^2\left(\rho(\beta_A-\underline{\beta}_A)-\sigma(\,^*\beta_A+\,^*\underline{\beta}_A)\right)+\tau_+^2\left(\alpha_{AB}\beta_B+\rho\beta_A-\sigma\,^* {\beta}_A\right)\right]{}^{(T)}\widehat{\pi}_{3A},
\end{align*}
and
\begin{align*}
I_5&=-\frac{1}{8}\left(\tau_-^2Q_{333A}+\tau_-^2Q_{343A}+\tau_+^2Q_{433A}+\tau_+^2Q_{443A}\right){}^{(T)}\widehat{\pi}_{4A}\\
&=-\frac{1}{2}\left[\tau_+^2\left(\rho(\beta_A-\underline{\beta}_A)-\sigma(\,^*\beta_A+\,^*\underline{\beta}_A)\right)-\tau_-^2\left(\underline{\alpha}_{AB}\underline{\beta}_B+\rho\underline{\beta}_A+\sigma\,^* \underline{\beta}_A\right)\right]{}^{(T)}\widehat{\pi}_{4A}.
\end{align*}
Therefore, the corresponding $\mathbf{I}_i$'s can be controlled exactly in the same way as above: we bound the curvature terms using \eqref{est: curvature} and we bound the deformation terms by \eqref{est: deformation tensor T Kb 2}. This leads to 
\[\left|\int_{t_{1}}^{t_{2}}\phi \left(\int_{ \{\frac{r_0}{2}\leq r \leq 2r_0\}}Q_{\alpha\beta\gamma\delta}{}^{(T)}\widehat{\pi}^{\alpha\beta}\overline{K}^{\gamma}T^{\delta}\right)ds\right|\lesssim \varepsilon_0^3 \log(t_2).\]

We can repeat the above argument to derive
\[\left|\int_{t_{1}}^{t_{2}}\phi \left(\int_{ \{\frac{r_0}{2}\leq r \leq 2r_0\}}Q_{\alpha\beta\gamma\delta}{}^{(K)}\widehat{\pi}^{\alpha\beta}\overline{T}^{\gamma}T^{\delta}\right)ds\right|\lesssim \varepsilon_0^3 \log(t_2).\]

Finally, by setting $t_{1}=0$ and $t_{2}=t$, we obtain that
\begin{equation}\label{conformal energy estimates}
\int_{\Sigma_{t}}Q (\overline{K},T,T,T) \lesssim \varepsilon_0^{2}+\varepsilon_0^{3}\log(1+t).
\end{equation}

To compute the limit of $E(t)$, we decompose it into two parts:
\[E(t)=\int_{\Sigma_{t}\cap \{ |r_{0}-r|\geq r_{0}^{\frac{2}{3}}\}}Q(T,T,T,T)+\int_{\Sigma_{t}\cap \{ |r_{0}-r|\geq r_{0}^{\frac{2}{3}}\}}Q(T,T,T,T).\]
In view of \eqref{eq: bound Q by conformal Q}, the first part can be bounded by the conformal energy estimate \eqref{conformal energy estimates}:
\begin{align*}
	\int_{\Sigma_{t}\cap \{ |r_{0}-r|\geq r_{0}^{\frac{2}{3}}\}}Q(T,T,T,T)&\lesssim \int_{\Sigma_t} \tau_+^{-2}Q (\overline{K},T,T,T)\\
	&\lesssim t^{-\frac{4}{3}}\left(\varepsilon_0^{2}+\varepsilon_0^{3}\log(1+t)\right).
\end{align*}
When $t\rightarrow \infty$, its contribution to the limit is $0$.

The second part deals with the near the light cone region. We will apply the improved decay on $\underline{\alpha}$. In view of the expression of $Q(T,T,T,T)$ in \eqref{eq: computation for Q}, the worst decay term  comes from $\underline{\alpha}$. Therefore, we have
\begin{align*}
	\int_{\Sigma_{t}\cap \{ |r_{0}-r|\leq r_{0}^{\frac{2}{3}}\}}Q(T,T,T,T)&\lesssim \int_{\Sigma_{t}\cap \{ |r_{0}-r|\leq r_{0}^{\frac{2}{3}}\}} r_{0}^{-3}\varepsilon_0^2.
\end{align*}
In view of Lemma \ref{cor:volume}, we bound the volume of the region $\Sigma_{t}\cap \big\{ |r_{0}-r|\leq r_{0}^{\frac{2}{3}}\big\}$ by a universal constant times $(r_{0}+r_{0}^{\frac{2}{3}})^{3}-(r_{0}-r_{0}^{\frac{2}{3}})^{3}\lesssim r_{0}^{\frac{8}{3}}$. Hence,
\begin{align*}
	\int_{\Sigma_{t}\cap \{ |r_{0}-r|\geq r_{0}^{\frac{2}{3}}\}}Q(T,T,T,T)&\lesssim (1+t)^{-\frac{1}{3}}\varepsilon_0^2.
\end{align*}
When $t\rightarrow \infty$, the contribution of this part to the limit is also $0$. We conclude that  $\lim_{t\rightarrow\infty}E(t)=0$. This completes the proof.
\end{proof}

We are ready to prove the main theorem of the paper. We define
\[E_*=\sup_{t\geq 0}E(t).\]
We choose a $t_0\in [0,\infty)$ so that  $E(t_0)\geq \frac{1}{2}E_*$. According to Lemma \ref{lemma: lemma5}, for all $t\geq t_0$, we have
\[
	|E(t)-E(t_{0})|\lesssim \varepsilon_0 \sup_{s\in [t_{0},t]}E(s) \leq  \varepsilon_0E_*.
\]
Since $\lim_{t\rightarrow\infty}E(t)=0$, when $t\rightarrow \infty$, this leads to 
\[
	 \frac{1}{2}E_* \leq E(t_{0})\lesssim  \varepsilon_0E_*.
\]
Therefore, for sufficiently small $\varepsilon_0$, we have $E_*=0$. Hence, $E(t)\equiv 0$. This shows that all the curvature components of the spacetime vanish, i.e., $(M,g)$ is flat.

\begin{Remark}\label{remark: using positive mass theorem}
There is another way to prove the theorem via the rigidity part of the positive mass theorem. We sketch a possible proof as follows:

The first ingredient is to consider the the radiation field $\underline{\mathbf{\Xi}}$ of $\widehat{\underline{\chi}}$, which is defined by 
\[\underline{\mathbf{\Xi}}:=-\frac{1}{2}\lim_{t\rightarrow 0, \atop \\ {\text on}~C_u}r\widehat{\underline{\chi}},
\]
See Chapter 17 of \cite{Christodoulou-Klainerman}. We can prove that $\underline{\mathbf{\Xi}}\equiv 0$ provided the radiation field $\underline{\mathbf{A}}=0$.

The second ingredient is to consider the Bondi mass loss formula which is also proved in Chapter 17 of \cite{Christodoulou-Klainerman}. The formula implies that the Bondi mass $M(u)$ satisfies the following 
\[\frac{d}{du}M(u)=\frac{1}{8\pi}\int_{\mathbf{S}^2}|\underline{\mathbf{\Xi}}|^2 d\mathring{\mu}\equiv 0,\]
where $d\mathring{\mu}$ the volume form on the standard unit sphere $\mathbf{S}^2$. On one hand, the Bondi mass vanishes at timelike infinity, i.e., $\lim_{u\rightarrow +\infty}M(u)=0$. On the other hand, $\lim_{u\rightarrow -\infty}M(u)$ is the ADM mass $m_{_{\small \rm ADM}}$. Hence, $m_{_{\small \rm ADM}}=0$. Therefore, we can apply the positive mass theorem. 
\end{Remark}


\begin{thebibliography}{99}	

\bibitem{BFJ} D. Fajman, J. Joudioux, J. Smulevici, \textit{The stability of the Minkowski space for the Einstein-Vlasov system}, Anal. PDE 14 (2021), no. 2, 425–531. 

\bibitem{Bieri-Zipser} L. Bieri, N. Zipser, \textit{Extensions of the stability theorem of the Minkowski space in general relativity}, AMS/IP Studies in Advanced Mathematics, vol. 45, American Mathematical Society, Providence, RI; International Press, Cambridge, MA, 2009. 

\bibitem{Choquet-Bruhat} Y. Choqu\'et-Bruhat, \textit{Th\'eor\`eme d'existence pour certains syst\`emes d’equations aux d\'eriv\'ees partielles nonlin\'eaires}, Acta Math. 88 (1952), 141-225.

\bibitem{Choquet-Bruhat-Geroch} Y. Choquet-Bruhat, R.P. Geroch, \textit{Global aspects of the Cauchy problem in General Relativity}, Commun. Math. Phys. 14, 329-335 (1969).


\bibitem{Christodoulou-Klainerman} D. Christodoulou, S. Klainerman, \textit{The global nonlinear stability of the Minkowski space}, Princeton Math. Ser. 41, Princeton University Press, Princeton, NJ, 1993.

\bibitem{Ionescu-Pausader} A.Ionescu, B.Pausader, \textit{The Einstein-Klein-Gordon coupled system: global stability of the Minkowski solution}, arXiv:1911.10652.

\bibitem{XWang} X. Wang, \textit{Global stability  of the Minkowski spacetime for the Einstein-Vlasov system},arxiv:2210.00512, preprint(2022).

\bibitem{Klainerman-Nicolo1} S. Klainerman, F. Nicol\`{o}, \textit{The evolution problem in general relativity}, Progress in Mathematical Physics, vol. 25, Birkh\"{a}user Boston, Inc., Boston, MA, 2003. 
	
\bibitem{Klainerman-Nicolo2} S. Klainerman, F. Nicol\`{o}, \textit{Peeling properties of asymptotically flat solutions to the Einstein vacuum equations}, Class. Quantum Grav. 20 (2003), 3215–3257.

\bibitem{Li-Yu} M. Li, P. Yu, \text{On the rigidity from infinity for nonlinear Alfvén waves}, J. Differential Equations 283 (2021), 163–215.

\bibitem{Lindblad} H. Lindblad, \textit{On the asymptotic behavior of solutions to the Einstein vacuum equations in wave coordinates}, Comm. Math. Phys. 353 (2017), no. 1, 135–184.

\bibitem{Lindblad-Rodnianski1} H. Lindblad, I. Rodnianski, \textit{Global existence for the Einstein vacuum equations in wave coordinates}, Comm. Math. Phys. 256 (2005), no. 1, 43–110.

\bibitem{Lindblad-Rodnianski2} H. Lindblad, I. Rodnianski, \textit{The global stability of Minkowski space-time in harmonic gauge}, Ann. of Math. (2) 171 (2010), no. 3, 1401–1477.

\bibitem{Lindblad-Taylor} H. Lindblad, M.Taylor, \textit{Global stability of Minkowski space for the Einstein-Vlasov system in the harmonic gauge}, Arch. Ration. Mech. Anal. 235 (2020), no. 1, 517–633. 

\bibitem{LeFloch-Ma} P. G. LeFloch, Y. Ma, \textit{The global nonlinear stability of Minkowski space for self-gravitating massive fields}, Comm. Math. Phys. 346 (2016), no. 2, 603–665.

\bibitem{Speck} J. Speck, \textit{The global stability of the Minkowski spacetime solution to the Einstein-nonlinear system in wave coordinates}, Anal. PDE 7 (2014), no. 4, 771–901.

\bibitem{Schoen-Yau1} R. Schoen, S.T. Yau,  \textit{On the proof of the positive mass conjecture in general relativity}, Comm. Math. Phys. 65 (1979), no. 1, 45–76.

\bibitem{Schoen-Yau2} R. Schoen, S.T. Yau, \textit{Proof of positive mass theorem. II}, Comm. Math. Phys. 79 (1981), no. 2, 231–260.


\bibitem{Wang} Q. Wang, \textit{An intrinsic hyperboloid approach for Einstein Klein-Gordon equations}, J. Differential Geom. 115 (2020), no. 1, 27–109.



\end{thebibliography}
\end{document}